\documentclass[12pt]{article}

\title{\color{navyblue} Geographic Difference-in-Discontinuities}
\author{\href{https://kylebutts.com/}{Kyle Butts}\thanks{University of Colorado Boulder, Economics Department (\href{mailto:kyle.butts@colorado.edu}{kyle.butts@colorado.edu})}}
\date{\footnotesize\today}

\usepackage[margin=1.25in]{geometry}

\usepackage{amsmath}
\usepackage{amsfonts}
\usepackage{amsthm}
\usepackage{graphicx}

\usepackage[T1]{fontenc}
\usepackage[utopia, varg]{newtxmath}

\usepackage{microtype}

\usepackage{tikz}

\usepackage{xcolor}

\definecolor{red}{HTML}{c62828}
\definecolor{orange}{HTML}{ef6c00}
\definecolor{green}{HTML}{2e7d32}
\definecolor{blue}{HTML}{1565c0}
\definecolor{purple}{HTML}{283593}
\definecolor{maroon}{HTML}{AF3335}
\definecolor{dark-maroon}{HTML}{5D0F0D}
\definecolor{teal}{HTML}{00695c}
\definecolor{bluegrey}{HTML}{455a64}
\definecolor{indigo}{HTML}{1A237E}
\definecolor{navyblue}{HTML}{0A3044}
\definecolor{bluegreen}{HTML}{4A8676}
\definecolor{Black}{HTML}{000000}

\definecolor{buff-gold}{HTML}{CFB87C}
\definecolor{buff-grey}{HTML}{565A5C}
\definecolor{buff-lightgrey}{HTML}{A2A4A3}
\definecolor{buff-black}{HTML}{000000}

\usepackage{hyperref}
\hypersetup{
    colorlinks= true,
    citecolor= dark-maroon,
    linkcolor= dark-maroon,
    filecolor= dark-maroon,      
    urlcolor= dark-maroon,
}

\usepackage{natbib}
\newtheoremstyle{spacing}
{}
{}
{}
{}
{\bfseries\color{navyblue}}
{.}
{2.5mm}
{}

\theoremstyle{spacing}
\newtheorem{theorem}{Theorem}

\newtheorem{assumption}{Assumption}

\global\long\def\one{\mathbf{1}}%

\usepackage{titling}

\pretitle{\begin{flushleft}\huge}
\posttitle{\end{flushleft}\vspace{-5mm}}
\preauthor{\begin{flushleft}\LARGE}
\postauthor{\end{flushleft}\vspace{-5mm}}
\predate{\begin{flushleft}\scshape\Large}
\postdate{\end{flushleft}\vspace{-5mm}}

\renewenvironment{abstract}
 {\noindent\rule{\linewidth}{.5pt}\noindent}
 {\noindent\rule{\linewidth}{.5pt}}

\usepackage[explicit]{titlesec}

\titleformat{\section}
  {\Large \bf \color{navyblue}}
  {\thesection \,---}
  {0.25em}
  {#1}
  
\titleformat{\subsection}
  {\fontsize{11}{10}\it}
  {\thesubsection.}
  {1em}
  {#1}

\let\oldfootnote\footnote
\renewcommand\footnote[1]{\oldfootnote{\ #1}}

\makeatletter
\renewcommand\@makefntext[1]{%
    \parindent 1em \noindent
    \hb@xt@1.8em{\hss\normalfont\@thefnmark.\hfill}#1
  }
\makeatother

\usepackage{enumitem}
\setitemize{labelindent=0.5em,labelsep=0.25cm,leftmargin=*}
\setenumerate{labelindent=0.5em,labelsep=0.25cm,leftmargin=*}

\usepackage{caption}

\DeclareCaptionLabelSeparator{threedash}{\,---\,}
\DeclareCaptionFont{navyblue}{\color{navyblue}}
\usepackage{subcaption}

\makeatletter
\let\input\@@input
\makeatother

\usepackage{adjustbox}

\usepackage{array}

\usepackage[flushleft]{threeparttable}
\usepackage{booktabs}

\usepackage{tabularx}
\newcolumntype{L}{>{\raggedright\arraybackslash}X}
\newcolumntype{R}{>{\raggedleft\arraybackslash}X}
\newcolumntype{C}{>{\centering\arraybackslash}X}

\usepackage{pdflscape}

\usepackage{fancyhdr}
\fancypagestyle{lscaped}{%
    \fancyhf{}
    
    \textnormal
    \fancyfoot{%
        \tikz[remember picture,overlay]
        \node[outer sep=2.5cm,above,rotate=90] at (current page.east) {\thepage};
    }
}

\hypersetup{
    pdftitle={Geographic Difference-in-Discontinuities},    
    pdfauthor={Kyle Butts}
}

\begin{document}

\begin{titlepage}
    \maketitle
    
    \begin{abstract}
        {\small
        A recent econometric literature has critiqued the use of regression discontinuities where administrative borders serves as the `cutoff'. Identification in this context is difficult since multiple treatments can change at the cutoff and individuals can easily sort on either side of the border. This note extends the difference-in-discontinuities framework discussed in \citet{Grembi_Nannicini_Troiano_2016} to a geographic setting. The paper formalizes the identifying assumptions in this context which will allow for the removal of time-invariant sorting and compound-treatments similar to the difference-in-differences methodology.
    
        \par
        }
    \end{abstract}
\end{titlepage}

\section{Introduction}

An increasingly popular estimation strategy involves using administrative borders as cutoffs in a regression discontinuity (RD) setting where the `running variable' is the distance to the border. The purpose of using observations close to the border `cutoff' is to try and better match treated and control units based on unobservable characteristics. Identification using the standard RD continuity assumption is problematic because many laws and institutions change discontinuously (i.e. compound treatment) at the border cutoff and people chose to sort on either side of the borders (i.e. sorting around cutoff), leading to important differences between units in close geographic proximity even in the counterfactual world without treatment.\footnote{Identification through randomization local to the cutoff does not make sense in the geographic context because that would require people to randomly be located on either side of the border.} 

The intuition of the difference-in-discontinuities design is very similar to the difference-in-differences design. A pre-treatment RD identifies time-invariant effects of other laws as well as the discontinuity in outcomes due to time-invariant sorting. A post-treatment RD identifies those previous two discontinuities plus the one caused by the treatment of interest. The difference between the two identifies the treatment effect. In this note, I extend the difference-in-discontinuities identification strategy formalized in \citet{Grembi_Nannicini_Troiano_2016} to the context of geographic RDs and discuss the particular identifying assumptions needed for the above identification sketch to be true when using a geographic RD. 

Figure \ref{fig:example} shows a stylized example of this identification. The left panel shows a discontinuity at the border cutoff that exists prior to treatment. This could be due to other policies changing at the border or sorting due to other reasons. The right panel shows the treated and untreated potential outcomes in the post-period. The key idea behind the diff-in-disc estimator is that the pre-period discontinuity can be estimated and removed from the second period discontinuity to estimate the treatment effect.

\begin{figure}[tb]
    \caption{Example of Diff-in-Disc Identification}
    \label{fig:example}

    \begin{adjustbox}{width = 1\textwidth, center}
        \includegraphics[width=\textwidth]{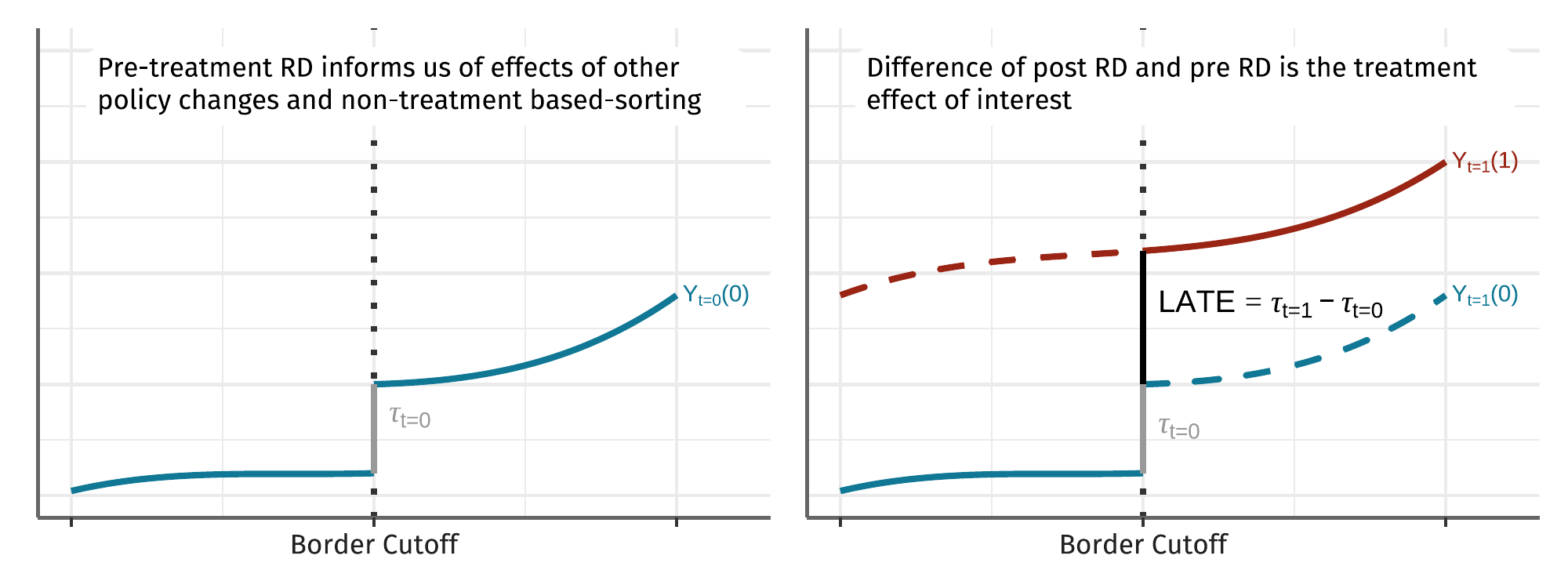}
    \end{adjustbox}

    {\footnotesize \emph{Notes:} This figure shows a stylized example of identification in the diff-in-disc setting.}
\end{figure}

I contribute to the econometric literature on RD in three ways. First, I contribute to the nascent literature formalizing the difference-in-discontinuities identification strategy \citep{Grembi_Nannicini_Troiano_2016,Galindo-Silva_Some_Tchuente_2021,Millan-Quijano_2020}. The results of the previous papers only consider the case of compound treatment where multiple treatment occurs at a cutoff. This paper formalizes the effectiveness of using difference-in-discontinuities to address the problem of time-invariant sorting around the cutoff.

The second contribution is that I extend on the work of \citet{Keele_Titiunik_2015} who formalize identification with geographic RDs into the \emph{geographic} difference-in-discontinuities setting. The authors raise the problem of sorting on either side of the boundary as well as multiple laws changing discontinuously at the boundary and propose stringent assumptions to avoid these problems in the cross-section. This paper uses the  difference-in-discontinuities methodology which provide a solution to these problems under arguably less stringent assumptions by leveraging the panel nature of data to estimate time-invariant sorting and effects of other policy/institution changes.

Last, I show that that estimation of difference-in-discontinuities with panel data can be done by running RD on outcomes that have been first-differenced. This allows for the use of modern advancements in estimation and inference of RDs.\footnote{See \citet{Cattaneo_Idrobo_Titiunik_2019} and \citet{Cattaneo_Idrobo_Titiunik} for an overview of modern techniques. The formulation using first-differences is practically useful as estimation can be done using the suite of RD packages found at \href{https://rdpackages.github.io/}{{https://rdpackages.github.io/}}.} In cases where panel data is not available, then the local regression framework proposed in \citet{Grembi_Nannicini_Troiano_2016} can be used.

\section{Methods}

\subsection{Traditional RD Identification}

Before introducing the difference-in-discontinuities method, I first review geographic RD to highlight difficulties in cross-sectional identification. I consider the standard context a random iid sample of units $i \in \{1, \dots n\}$. There is a running variable $D_i$ that measures distance to the border of a treated area. Without loss of generality, the distance is normalized to zero with positive distances being within the treatment area.\footnote{\citet{Keele_Titiunik_2015} discuss the choice of using a single measure of distance versus a two-dimensional running variable. The difference-in-discontinuity method can be extended into the two-dimensional framework easily, but data will usually render the two-dimensional case implausible.} The observed outcome is modeled by 
$$
    y_i = f(D_i) + \tau(D_i) \one(D_i \geq 0) + \underbrace{X_i \beta + u_i}_{\equiv \varepsilon_i}. 
$$
The function $f(D)$ summarizes location-specific characteristics that affect outcomes. For example, one variable could be proximity to a city and $f(D)$ summarizes its effect on the outcome variable. More generally, $f(D)$ captures amenities and labor markets as they change across space. On the other hand, $\varepsilon_i$ represents the \emph{potentially} unobserved individual-specific characteristics that affect outcome variable. The quantity $Y_i(0) = f(D_i) + \varepsilon_i$ determines the outcome variable in the absense of treatment and $\tau(D_i)$ is the average treatment effect at distance $D_i$. 

Identification of the treatment effect relies on the assumption that location-specific and individual-specific characteristics evolve smoothly across the border:

\begin{assumption}[RD]\label{eq:continuity}\ \\\vspace{-10mm} 
    \begin{itemize}
        \item[(i)] The functions $f(D)$ and $\tau(D)$ are continuous at the cutoff, $D = 0$,
        \item[(ii)] $\mathbb{E}\left[ \varepsilon_i \vert D_i = D \right]$ is continuous at the cutoff, $D = 0$. 
    \end{itemize}
\end{assumption}

Part (i) of (\nameref{eq:continuity}) says that the effect of the running variable on the outcome with and without treatment is continuous at the cutoff and part (ii) says that the effect of other \emph{potentially} unobserved individual-specific variables on the outcome are continuous at the cutoff. In the context of geographic discontinuities, a discontinuity in $f(D)$ could arise from multiple policies changing at the border and a discontinuity in $\mathbb{E}\left[ \varepsilon_i \vert D_i = D \right]$ could arise from sorting across the border \citep{Keele_Titiunik_2015}. These two problemsrepresent a central threat to identification of treatment effects in the geographic RD setting. 

If the two continuity assumptions are satisfied, observations in the control area close to the border identify the limiting value of $f(0)$ and observations in the treated area close to the border identify the limiting value of $f(0) + \tau(0)$. The difference between these two limits identify $\tau(0)$. Formally, for a variable $z$, the left and right limits at the cutoff are denoted $z^+ \equiv \lim_{D_i \to 0^+} z_i$ and $z^- \equiv \lim_{D_i \to 0^-} z_i$. With (\nameref{eq:continuity}), it is easy to show that the RD estimate identifies the treatment effect, i.e. $\tau = y^+ - y^-$.\footnote{See Theorem 1 in \citet{Hahn_Todd_Klaauw_2001}.} 

\begin{theorem}[RD Identification]
    Under (\nameref{eq:continuity}), $\tau(0) = y^+ - y^-$.
\end{theorem}


\subsection{Difference-in-Discontinuities Identification}

Now we turn to the panel setting where we observe outcomes before and after treatment occurs, $t \in \{0,1\}$. We observe a random iid panel sample of $\left\{ (y_{i0}, y_{i1}, D_i) \right\}_{i=1}^n$. In this setting, discontinuities at the border before treatment informs us on the affects of other treatments and time-invariant sorting. Potential outcomes for individual $i$ at time $t$ are now modeled as 
\begin{equation}\label{eq:panel_model}
    y_{it} = f_t(D_i) + \gamma(D_i) \one(D_i \geq 0) + \tau(D_i) \one(D_i \geq 0) \one(t = 1) + \varepsilon_{it},
\end{equation}
where $\gamma(D)$ represent the time-invariant discontinuity at the cutoff which could be due to time-invariant sorting and/or the effects of other policies that change at the border; the untreated location-specific component, $f_t(D)$, can very across periods; and $\tau(D)$ remains the treatment effect of interest. 

The assumptions necessary to identify the treatment effect $\tau(0)$ requires the traditional RD assumptions to hold in both periods.

\begin{assumption}[Diff-in-Disc]\label{eq:continuity_panel}\ \\\vspace{-10mm} 
    \begin{itemize}
        \item[(i)] The functions $f_0(D), f_1(D),$ and $\tau(D)$ are continuous at $D = 0$, 
        \item[(ii)] $\mathbb{E}\left[ \varepsilon_{it} \vert D_i = D \right]$ is continuous at the cutoff, $D = 0$, for $t \in \{0, 1\}$. 
    \end{itemize}
\end{assumption}

These assumptions warrant a bit of discussion. Note that the continuity assumption on $f_0$ is mild because discontinuities from compound treatment and time-invariatn sorting are allowed in $\gamma(D)$. This is an improvement over traditional geographic RDs which require these effects to not be present \citep{Keele_Titiunik_2015}.\footnote{This is similar to the difference-in-differences methodology that allows for time-invariant differences in levels.} In the post period, (\nameref{eq:continuity_panel}) requires two things. First, $f_1(D)$ being continuous requires that no other policies turn on between periods that would cause a discontinuity at the border. Second, it requires that the effects of previous policies were already fully developed in period 0. If the effects of other policies change over time, then the changes in effects over time would not be absorbed by $\gamma(D)$ and would cause a discontinuity in $f_1(D)$ that would be mistaken as the treatment effect. Second, part (ii) requires that no additional sorting can occur between 0 and 1, whether that be due to treatment or lagged sorting from previous treatments.

To help with estimation of the treatment effect, we can reformulate our potential outcomes in a first-difference model, $$
    (y_1 - y_0) = (f_1(D_i) - f_0(D_i)) + \tau(D_i)\one(D_i \geq 0) + (\varepsilon_{i1} - \varepsilon_{i0}),
$$
where $\gamma(D_i)$ cancels out because it is time-invariant. 

\begin{theorem}[Diff-in-Disc Identification]
    Under (\nameref{eq:continuity_panel}), $\tau(0) = (y_1 - y_0)^+ - (y_1 - y_0)^-$.
\end{theorem}

\begin{proof}
    \begin{align*}
        (y_1 - y_0)^+ - (y_1 - y_0)^- &= \tau(D)^{+} + (f_1 - f_0)^+  + (\varepsilon_{i1} - \varepsilon_{i0})^+ - ((f_1 - f_0)^-  + (\varepsilon_{i1} - \varepsilon_{i0})^-) \\
        &= \tau(0) + (f_1^+ - f_1^-) + (f_0^+ - f_0^-) + (\varepsilon_1^+ - \varepsilon_1^-) + (\varepsilon_0^+ - \varepsilon_0^-) \\
        &= \tau(0),
    \end{align*}
    where the second equality comes from continuity of $\tau(D)$ and the last equality comes from the two continuity assumptions.
\end{proof}

The above theorem says that so long as sorting and other policies are fully observed in the pre-period, a regression discontinuity estimated on a first-differenced outcome will identify the treatment effect. This theorem is closely related to \citet{Grembi_Nannicini_Troiano_2016} but differs in an important way. First, they find that $\tau(0) = (y_1^+ - y_1^-) - (y_0^+ - y_0^-)$ which does not require panel data. In cases of panel data, formulating the above identification result in terms of first differences is advantageous. Since $(y_1 - y_0)^+ - (y_1 - y_0)^-$ is a standard RD estimate of the difference between the right and left limits, this  unlocks the wide set of econometric tools used in RD estimation including local polynomial regression, data-driven bandwidth selection, and bias-corrected inference. \citet{Cattaneo_Idrobo_Titiunik_2019} and \citet{Cattaneo_Idrobo_Titiunik} provide a literature review of the modern RD literature and include a set of R and Stata programs containing powerful estimation tools.

In the non-panel case, estimation can proceed in a local-polynomial regression framework as proposed by \citet{Grembi_Nannicini_Troiano_2016}. They recommend running the following regression using observations within a small interval around $D_i = 0$: 
\begin{align*}
    Y_{it} &= \delta_0 + \delta_1 * D_{i} + \one(D_i \geq 0) \left( \gamma_0 + \gamma_1 D_i \right) + \\ 
    &\quad \one(t = 1) \left( \alpha_0 + \alpha_1 * D_{i} + \one(D_i \geq 0) \left( \beta_0 + \beta_1 D_i \right) \right] + \eta_{it}.
\end{align*}
From standard regression results, $\beta_0$, will be the difference-in-discontinuities estimate. This estimation strategy, however, does not as easily allow for the use of modern bias-robust estimators.

\section{Discussion}

This paper extended the difference-in-discontinuities framework proposed by \citet{Grembi_Nannicini_Troiano_2016} into the context of geographic discontinuities. This setting faces the same problem of compound treatment that other RD contexts exhibit and since individuals can sort sort across the border, this context provides additional difficulties. This paper formalizes the necessary assumptions in the \emph{geographic} context in order to identify the treatment effect of a policy. Moreover, in the presence of panel data, this paper proposes improved estimation techniques by recasting the estimator as a RD estimator on first-differenced data.

\setlength{\bibsep}{0.0pt}
\bibliography{references.bib}

\end{document}